\documentclass[a4paper, 11pt]{article}

\usepackage[linkcolor=black,colorlinks=true,citecolor=black,filecolor=black]{hyperref}

\usepackage{amsmath,amssymb,amsthm}

\usepackage[linkcolor=black,colorlinks=true,citecolor=black,filecolor=black]{hyperref}

\usepackage{graphicx,subfigure}
\usepackage[a4paper, margin=1.2in]{geometry}

\usepackage{algorithm,algorithmic}

\newtheorem{theorem}{Theorem}

\newtheorem{lemma}[theorem]{Lemma}

\newtheorem{definition}[theorem]{Definition}

\newtheorem{observation}[theorem]{Observation}

\newcommand{\E}{\ensuremath{\mathbf{E}}}

\newcommand{\sat}{\mathrm{sat}}
\newcommand{\ol}[1]{\bar{#1}}
\newcommand{\psuccess}{p_{\mathrm{success}}}
\newcommand{\pguessed}{p_{\mathrm{guessed}}}
\newcommand{\PPSZ}{\textsc{PPSZ}}
\newcommand{\SL}{\mathrm{SL}}
\newcommand{\VC}{V_\mathrm{F}}
\newcommand{\VN}{V_\mathrm{N}}
\newcommand{\nC}{n_\mathrm{F}}
\newcommand{\nN}{n_\mathrm{N}}
\newcommand{\boldF}{\mathbf{F}}
\newcommand{\boldG}{\mathbf{G}}
\newcommand{\AssignSatisfiableLiterals}{\mathrm{AssignSL}}

\begin{document}

\title{$3$-SAT Faster and Simpler - Unique-SAT Bounds for PPSZ Hold in General}


\author{
Timon Hertli\\
\texttt{timon.hertli@inf.ethz.ch}\\
\vspace*{3mm}\\
Institute for Theoretical Computer Science\\
Department of Computer Science\\
ETH Z\"urich, 8092 Z\"urich, Switzerland
}
\maketitle

\begin{abstract}
The PPSZ algorithm by Paturi, Pudl\'ak, Saks, and Zane~\cite{ppsz} is the fastest known algorithm for Unique $k$-SAT, where the input formula does not have more than one satisfying assignment. For $k\geq 5$ the same bounds hold for general $k$-SAT. We show that this is also the case for $k=3,4$, using a slightly modified PPSZ algorithm. We do the analysis by defining a \emph{cost} for satisfiable CNF formulas, which we prove to decrease in each PPSZ step by a certain amount. This improves our previous best bounds with Moser and Scheder~\cite{hms10} for $3$-SAT to $O(1.308^n)$ and for $4$-SAT to $O(1.469^n)$.
\end{abstract}

\section{Introduction}
$k$-SAT and especially $3$-SAT is one of the most prominent NP-complete problems. While a polynomial algorithm seems very unlikely, much effort has been put into finding ``moderately exponential algorithms'', i.e.\ algorithms running in time $O(c^n)$ for $c<2$, where $n$ denotes the number of variables of the input formula. 
Unique $k$-SAT is the variant of the $k$-SAT problem where the input CNF formula is promised to have a unique or no satisfying assignment.
In 1998, Paturi, Pudl\'ak, Saks, and Zane~\cite{ppsz} presented a randomized algorithm for Unique $3$-SAT that runs in time $O(1.30704^n)$, where $n$ is the number of variables of the formula. For general $3$-SAT, only a running time of $O(1.3633^n)$ could be shown using a complicated analysis. Shortly afterwards, Sch\"oning~\cite{schoning1999} proposed a very simple algorithm with running time $O(1.33334^n)$ for $3$-SAT. In 2004, Iwama and Tamaki~\cite{it04} showed that Sch\"oning's algorithm can be combined with PPSZ to get a running time of $O(1.32373^n)$\footnote{Using the new version of~\cite{ppsz} immediately gives the bound $O(1.32267^{n})$, as stated in~\cite{rolf2006}.}. This bound was subsequently improved by Rolf~\cite{rolf2006}, Iwama, Seto, Takai, and Tamaki~\cite{istt10} and in~\cite{hms10} with Moser and Scheder to $O(1.32216^n)$, $O(1.32113^n)$, $O(1.32065^n)$, respectively. However, these bounds are still far from the bound $O(1.30704^n)$ for Unique 3-SAT.

PPSZ~\cite{ppsz} does the following: First, the input formula $\boldF$ is preprocessed by $s$-bounded resolution, meaning that all clauses obtainable by resolution when clauses of size at most $s$ are considered are added to $\boldF$. Then PPSZ goes through the variables in random order. In each step, a variable $x$ is permanently replaced by a Boolean value $a$ as follows: If there  a clause $\{x\}$ or $\{\ol{x}\}$ in the current formula, then $a$ is chosen accordingly and we call $x$ \emph{forced}. Otherwise $a$ is chosen uniformly at random from $\{0,1\}$ and we call $x$ \emph{guessed}. In Unique $k$-SAT it was shown that the probability that a variable is guessed is bounded from above by some quantity $S$ depending on $k$ on $\epsilon$. Using this, it is not hard to show that there exists an algorithm for Unique $k$-SAT running in time $O(2^{S\cdot n})$.
\subsection{Our Contribution}
We give an analysis of a slightly adapted PPSZ algorithm that achieves the same bound for $k$-SAT as for Unique $k$-SAT, which gives new bounds for $k=3,4$. The previously best known bounds from~\cite{hms10} are improved for $3$-SAT from $O(1.32065^n)$ to $O(1.30704^n)$ and for $4$-SAT from $O(1.46928^n)$ to $O(1.46899^n)$.
\begin{theorem}
\label{bound-3sat}
There exists a randomized algorithm for $3$-SAT with one-sided error that runs in time $O(1.30704^n)$.
\end{theorem}

\begin{theorem}
\label{bound-4sat}
There exists a randomized algorithm for $4$-SAT with one-sided error that runs in time $O(1.46899^n)$.
\end{theorem}

Our analysis is directly based on the analysis for Unique $k$-SAT of~\cite{ppsz}, we do not use the part that considers general $k$-SAT. Let $\boldF$ be a satisfiable CNF formula. A variable $x$ of $F$ is called \emph{frozen}\footnote{We previously called frozen variables \emph{critical}. Such variables are also referred to as \emph{backbone}. Thanks to Ramamohan Paturi for making us aware of the existing terminology.} if it has the same value in all satisfying assignments, and \emph{non-frozen} otherwise. In~\cite{hms10} it was shown how to use this distinction to improve PPSZ: A frozen variable is good for PPSZ, while assigning any Boolean value to a non-frozen variable preserves satisfiability. Using a preprocessing step, about half of the variables can be assumed to be frozen for PPSZ. In this paper, we use a similar but more fine-grained approach:

We assign to each $k$-CNF formula $\boldF$ a \emph{cost} $c(\boldF)\leq S\cdot n$, where $S$ is the upper bound that a variable is guessed in Unique $k$-SAT. The main theorem then shows that $\PPSZ$ finds a satisfying assignment of $\boldF$ with probability at least $2^{-c(\boldF)}$. By a routine argument this gives a randomized algorithm with running time $O(2^{S\cdot n})$. Using the cost function, we are able to consider just one PPSZ step. We show that the cost decreases, depending on how many variables are frozen. If all variables are frozen, the cost decreases by $1$. If some variables are non-frozen, the decrease is smaller, but due to the presence of non-frozen variables, satisfiability is preserved with higher probability in this step. With this, the theorem can be obtained by induction.

In Section \ref{section-ppsz}, we review the PPSZ algorithm and prove its properties we need. In Section \ref{section-cost}, we introduce the cost function and do the analysis using the statements of Section \ref{section-ppsz}.
\subsection{Notation}
We adapt the notational framework as used in \cite{welzl05}. 
Let $V$ be a finite set of propositional \emph{variables}. A \emph{literal} $l$ over $x\in V$ is a variable $x$ or a complemented variable $\ol{x}$. If $l=\ol{x}$, then $\ol{l}$, the complement of $l$, is defined as $x$.
We assume that all lite\-rals are distinct. A \emph{clause} over $V$ is a finite set of lite\-rals over pairwise distinct variables from $V$.
A formula in \emph{CNF} (Conjunctive Normal Form) is a pair $\boldF=(F,V)$ where $F$ is a finite set of clauses over $V$. $V$ is the set of variables of $\boldF$ and denoted by $V(\boldF)$. We define $n(\boldF):=|V(\boldF)|$, the number of variables of $\boldF$. If $\boldF$ is understood from the context, we sometimes write $n$ for $n(\boldF)$.
A clause containing exactly one literal is called a \emph{unit clause}. We
say that $\boldF=(F,V)$ is a 
$(\le k)$\emph{-CNF} formula
if every clause of $F$ has
size at most $k$.

Let $V$ be a finite set of variables.
A (truth) \emph{assignment} on $V$ is a function $\alpha : V \rightarrow
\{0,1\}$ which assigns a Boolean value to each variable. A literal
$u=x$ (or $u=\ol{x}$) is \emph{satisfied by} $\alpha$ if $\alpha(x)=1$
(or $\alpha(x)=0$). A clause is \emph{satisfied by} $\alpha$ if it
contains a satisfied literal and a formula is \emph{satisfied by}
$\alpha$ if all of its clauses are. A formula is \emph{satisfiable} if
there exists a satisfying truth assignment to its variables.
Given a CNF formula $\boldF$, we denote by $\sat(\boldF)$ the set of assignments on $V(\boldF)$ that
satisfy $\boldF$. $k$\emph{-SAT} is the decision problem of deciding if a $(\leq k)$-CNF formula has a satisfying assignment.

Formulas can be manipulated by permanently assigning values to
variables. If $\boldF$ is a given CNF formula and $x \in V(\boldF)$
then assigning $x \mapsto 1$ satisfies all clauses containing $x$
(irrespective of what values the other variables in those clauses are
possibly assigned later) whilst it truncates all clauses containing
$\ol{x}$ to their remaining literals. Additionally, $x$ is removed from the variable set of $\boldF$.
We will write $\boldF^{[x\mapsto1]}$ (and analogously $\boldF^{[x\mapsto 0]}$) to
denote the formula arising from doing just this. 
%
For notational convenience, we also write $x$ for $x\mapsto 1$ and $\ol{x}$ for $x\mapsto 0$, i.e.\ we write a literal instead of a variable-value pair. With this, we can view an assignment $\alpha$ also as the set of literals $l$ that are satisfied by $\alpha$. If the literal $l$ corresponds to $x\mapsto a$, we write $\boldF^{[l]}$ instead of $\boldF^{[x\mapsto a]}$.
%
%
By choosing an element from a finite set u.a.r., we mean choosing it
uniformly at random. 
Unless otherwise stated, all random
choices are mutually independent.
We denote by $\log$ the logarithm to the base 2. For the logarithm to
the base $e$, we write $\ln$. 
\section{The PPSZ Algorithm}
\label{section-ppsz}
In this section we present and slightly modify the PPSZ algorithm from~\cite{ppsz}. We also introduce the concept of \emph{frozen variables} from~\cite{hms10} (called critical variables there) and present the statements about the PPSZ algorithm we need later. To analyze PPSZ, we can ignore unsatisfiable formulas, as in that case PPSZ never returns a satisfying assignment. In the rest of this paper we fix an integer $k\geq 3$ and let $\boldF=(F,V)$ be an arbitrary satisfiable $(\leq k)$-CNF.
For our analysis, we need to change the PPSZ algorithm slightly. The PPSZ algorithm was defined using a preprocessing step of $s$-bounded resolution, i.e.\ resolution when only considering clauses of size at most $s$. We change this to a weaker concept we call $s$-implication\footnote{The concept of $s$-implication is from the lecture note draft by Dominik Scheder.}. 
We call a literal $s$-implied if it is implied by a subformula with at most $s$ clauses:
\begin{definition}
Let $\boldF=(F,V)$ be a satisfiable CNF formula. We say that a literal $l$ is \emph{$s$-implied} by $\boldF$ if there is a subset $G$ of $F$ with $|G|\leq s$ such that all satisfying assignments of $\boldG=(G,V)$ set $l$ to $1$. For notational convenience, we also say that a variable $x$ is $s$-implied, if one of the literals $x$ or $\ol{x}$ is $s$-implied.

We call a CNF formula $\boldF$ \emph{$s$-implication free} if no literal $l$ is $s$-implied.
\end{definition}

\begin{algorithm}
\caption{$\PPSZ($CNF formula $\boldF$, integer $s$)}
\begin{algorithmic}
  \STATE Choose $\beta$ u.a.r.\ from all assignments on $V(\boldF)$ 
  \STATE Choose $\pi$ u.a.r.\ from all permutations of $V(\boldF)$ 
  \RETURN $\PPSZ(\boldF,\beta, \pi, s)$
\end{algorithmic}
\end{algorithm}
\begin{algorithm}
\caption{$\PPSZ($CNF formula $\boldF$, assignment $\beta$, permutation $\pi$, integer $s$)}
\begin{algorithmic}
\STATE $V\gets V(\boldF)$
\STATE Let $\alpha$ be a 
partial assignment over $V$, initially the empty assignment
\FOR {all $x \in V$, according to $\pi$}
\WHILE {there is an $s$-implied literal $l=y\mapsto a$ in $\boldF$}
\STATE $\boldF\gets \boldF^{[y\mapsto a]}$
\STATE $\alpha(y)\gets a$
\ENDWHILE
\IF {$x\in V(\boldF)$}
\STATE $\boldF\gets \boldF^{[x\mapsto \beta(x)]}$
\STATE $\alpha(x)\gets \beta(x)$
\ENDIF
\ENDFOR
\RETURN $\alpha$
\end{algorithmic}
\end{algorithm}

Our analysis requires the following modification of PPSZ: Instead of processing the variables strictly step by step, we check after each step for which variables we know the value (by $s$-implication) and immediately set these accordingly. While the change to $s$-implication makes PPSZ only weaker, this modification makes the algorithm stronger; our approach does not work with the PPSZ algorithm proposed in~\cite{ppsz}, more on this is written in the conclusion.
In the following we fix $s$ large enough for the bounds we want to show, as described later. It is easily seen that the $\PPSZ$ algorithm we present runs in polynomial time if $s$ is a constant. 
We now give some definitions used in the analysis:
\begin{definition}[\cite{ppsz}]
Let $\beta$ and $\pi$ be chosen randomly as in $\PPSZ(\boldF,s)$.
We define the \emph{success probability} of PPSZ as the probability that it returns a satisfying assignment.
\[\psuccess(\boldF,s):=\Pr_{\pi,\beta}\left(\PPSZ(\boldF,\beta,\pi)\in \sat(\boldF)\right).\]
Consider a run of $\PPSZ(\boldF,s)$. For $x\in V(\boldF)$ we call $x$ \emph{forced} if the value of $x$ is determined by $s$-implication; we call it \emph{guessed} otherwise.
For $\alpha\in \sat(\boldF)$ we define the probability that $x$ is guessed w.r.t.\ $\alpha$ as follows:
\[\pguessed(\boldF,x,\alpha,s):=\Pr_\pi\left(x\textrm{ is guessed in }\PPSZ(\boldF,\alpha,\pi,s)\right).\]
To make notation easier, we extend the notation and allow $\alpha$ also to be over a variable set $W\supset V(\boldF)$; the variables in $W\setminus V(\boldF)$ will then be ignored. We also allow $x\not\in V(\boldF)$ and define in this case $\pguessed(\boldF,x,\alpha,s):=0$.
\end{definition}
\begin{definition}[\cite{hms10}]
We say that $x\in V(\boldF)$ is \emph{frozen} if all satisfiable assingments of $\boldF$ agree on $x$. We say that $x$ is \emph{non-frozen} otherwise.
\end{definition}

The probability that a frozen variable is guessed can be bounded:
\begin{theorem}[\cite{ppsz,hms10}]
\label{theorem-guessed-bound}
If $x$ is a frozen variable, then $\pguessed(\boldF,x,\alpha,s)\leq S_k + \epsilon_k(s)$, where $S_k:=\int_0^1 \frac{t^{1/(k-1)}-t}{1-t}dt$ and $\epsilon(s)$ goes to $0$ for $s\to\infty$.
\end{theorem}
\begin{proof}[Proof Sketch]
If there is a unique (or sufficiently isolated) satisfying assignment, then~\cite{ppsz} gives us an upper bound for the probability that a variable is guessed. 
 In~\cite{hms10}, we have showed that this bound also holds for an arbitrary satisfying assignment, as long as the variable is \emph{frozen}. It is easily seen that this bound holds if we use $s$-implication instead of a preprocessing step of $s$-bounded resolution: In the analysis of~\cite{ppsz}, so-called critical clause trees are used to bound $\pguessed$. The only clauses $D$ used in the proof there are these with a resolution deduction using at most $m$ clauses of $\boldF$ with size at most $k$ each, for some appropriately chosen constant $m$. Then $D$ can be obtained by $(m\cdot k)$-bounded resolution from $\boldF$. This also means that $D$ is implied by at most $m$ clauses. If we restrict $\boldF$ to some literals and obtain $\boldF'$, then the clause $D$ restricted to these literals is now implied by at most $m$ clauses of $\boldF'$. Hence appearance of a unit clause in the algorithm of~\cite{ppsz} now becomes $s$-implication of a literal here for all unit clauses considered in the analysis of~\cite{ppsz}.
\end{proof}
For $k=3$, we can show that $S_3=2\ln 2 - 1$. For small $k$, $S_k$ and $2^{S_k}$ are approximately (rounded up):
\begin{tabular}{c|c|c}
$k$&$S_k$&$2^{S_k}$\\
\hline
3&0.3862944&1.307032\\
4&0.5548182&1.468984\\
5&0.6502379&1.569427\\
6&0.7118243&1.637874\\
\end{tabular}

In~\cite{ppsz} it was shown using Jensen's inequality how to use a bound for the probability that a variable is guessed to give an upper bound for the running time of Unique $k$-SAT. For $k\geq 5$, the same bound holds also for $k$-SAT using a more elaborate argument.
\begin{theorem}[\cite{ppsz}]
For $\epsilon>0$, there exists a randomized algorithm for Unique $k$-SAT with one-sided error that runs in time $O(2^{S_k n+ \epsilon n})$. For $k\geq 5$, this is also true for $k$-SAT.
\end{theorem}
In this paper we prove that for $k$-SAT we have the same bound as Unique $k$-SAT for \emph{all} $k$. For $k\geq 5$, this can be seen as an alternative proof for general $k$-SAT. Note however that because we immediately fix implied variables, the algorithm is slightly different. In the remainder of this section, we will introduce additional notation and state some properties of PPSZ we will need later.
\begin{definition}
We denote the \emph{non-frozen} variables of $\boldF$ by $\VN(\boldF)$ and set $\nN(\boldF):=|\VN(\boldF)|$. We denote the \emph{frozen} variables by $\VC(\boldF)$ and set $\nC(\boldF):=|\VC(\boldF)|$. The \emph{satisfying literals}, denoted by $\SL(\boldF)$, are the literals $l$ over $V(\boldF)$ s.t.\ $\boldF^{[l]}$ is satisfiable.
\end{definition}
The satisfying literals consist of all literals over non-frozen variables, and for each frozen variable of the literal that corresponds to the satisfying assignments of $\boldF$. It follows that $|\SL(\boldF)|=2\nN(\boldF)+\nC(\boldF)$. 

The following alternative definition of $\PPSZ(\boldF,s)$ is easily seen to be the same algorithm. We will use this later to bound $\psuccess(\boldF,s)$.
\begin{observation}
\label{observation-induction}
We can alternatively characterize $\PPSZ(\boldF,s)$ as follows: 
We first set all $s$-implied literals in $\boldF$ accordingly, and let $\alpha$ be the assignment consisting of these literals. $\boldF$ is now $s$-implication free. If $n(\boldF)=0$, then we return $\alpha$. Otherwise we choose $x$ from $V(\boldF)$ u.a.r.\ and $a$ from $\{0,1\}$ u.a.r.\ and let $l:=x\mapsto a$. Then we run $\PPSZ(\boldF^{[l]},s)$ and combine the returned assignment with $\alpha\cup\{l\}$.

It follows that if $\boldF$ is $s$-implication free and if $n(\boldF)\geq 1$, then
\[\psuccess(\boldF,s)=\frac{1}{2n}\sum_{l\in\SL(\boldF)}\psuccess(\boldF^{[l]},s).\]
\end{observation}

We need two statements about $\pguessed$ for our proof. The first tells us that if we restrict $\boldF$ to a literal of $\alpha$ the probability that $x$ is guessed w.r.t.\ $\alpha$ cannot increase.
\begin{lemma}
\label{lemma-pguessed-monotone}
For $l\in\alpha$ and $\alpha\in\sat(\boldF)$, we have $\pguessed(\boldF^{[l]},x,\alpha,s)\leq \pguessed(\boldF,x,\alpha,s)$.
\end{lemma}
\begin{proof}
Let $l=y\mapsto \alpha(y)$. Assume $x$ is $s$-implied in $\PPSZ(F,\alpha,\pi,s)$. Let $\pi'$ be the permutation obtained by removing $y$ from $\pi$. We claim that $x$ is $s$-implied in $\PPSZ(F^{[l]},\alpha,\pi',s)$: Consider the clause set $G$ that $s$-implies $x$ in $\PPSZ(F,\alpha,\pi,s)$. It follows from the definition of $s$-implication that restricting $G$ to $l$ gives a clause set that $s$-implies $x$, and hence $x$ is $s$-implied in $\PPSZ(F^{[l]},\alpha,\pi',s)$. The statement is now easily seen, as $\pi'$ has the distribution of a permutation uniformly at random chosen from all permutations on $V(F)\setminus\{y\}$.
\end{proof}
The second statement allows us to relate the probability that a variable $x$ of $\boldF$ is guesssed to the probability that $x$ is guessed if $\boldF$ is restricted by a random literal of $\alpha$. Intuitively, assume we have an upper bound for $\pguessed$. With some probability, $x$ is guessed right now. Hence if $x$ is not guessed right now, the probability to be guessed in the remainder must slightly decrease.
\begin{lemma}
\label{lemma-pguessed-reduced}
For $\alpha\in\sat(\boldF)$ and $x\in V(\boldF)$ s.t.\ $x$ is not $s$-implied, we have
\[\pguessed(\boldF,x,\alpha,s)-\frac{1}{n}=\frac{1}{n(\boldF)}\sum_{l\in\alpha}\pguessed(\boldF^{[l]},x,\alpha,s).\]
\end{lemma}
\begin{proof}
Let $\pi$ be a random permutation on $V(\boldF)$ and let $y$ be the variable that comes first in $\pi$. We have by definition
\[\pguessed(\boldF,x,\alpha,s)=\Pr_\pi\left(x\textrm{ is guessed in }\PPSZ(\boldF,\alpha,\pi,s)\right).\]
By the law of total probability, this is
\[=\E_{y}\left[\Pr_{\pi}\left(x\textrm{ is guessed in }\PPSZ(\boldF,\alpha,\pi,s)\ |\ y\textrm{ comes first in }\pi\right)\right].\]
If $x=y$, then $x$ is always guessed, as $x$ is not $s$-implied. If $x\not=y$, then the probability under the expectation is easily seen to be $\pguessed(\boldF^{[y\mapsto\alpha(y)]},x,\alpha,s)$. Writing the expectation as a sum and using that $\pguessed(\boldF^{[y\mapsto\alpha(y)]},y,\alpha,s)$ is defined as $0$ gives us
\[\pguessed(\boldF,x,\alpha,s)=\Pr_{y}[x=y]1+\Pr_{y}[x\not= y]\frac{1}{n(\boldF)-1}\sum_{l\in\alpha}\pguessed(\boldF^{[l]},x,\alpha,s).\]
Trivially $\Pr_{y}[x=y]=\frac{1}{n(\boldF)}$ and $\Pr_{y}[x\not= y]=\frac{n(\boldF)-1}{n(\boldF)}$; the statement follows now easily.
\end{proof}
\section{Analysis using a Cost Function}
\label{section-cost}
To define the cost function, we first need to give a probability distribution on the set of all satisfying assignments of a CNF formula. We do this by defining a random process that repeatedly picks a satisfying literal:
\begin{definition}
We define the random process $\AssignSatisfiableLiterals(\boldF)$ that produces an assignment on $V(\boldF)$ as follows: Start with the empty assignment $\alpha$, and repeat the following step until $V(\boldF)=\emptyset$: Choose a satisfying literal $l\in\SL(\boldF)$ and add $l$ to $\alpha$; then let $\boldF\gets\boldF^{[l]}$. At the end, output $\alpha$.

Let $\alpha$ be an assignment on $V(F)$. Then $p(\boldF,\alpha)$ is defined as the probability that $\AssignSatisfiableLiterals(\boldF)$ returns $\alpha$. If $\alpha$ is defined on some $W\supset V(F)$, $p(\boldF,\alpha)$ is defined as the probability that $\AssignSatisfiableLiterals(\boldF)$ returns $\alpha$ restricted to $V(F)$.
\end{definition}
From the definition we observe the following:
\begin{observation}
\label{observation-distribution}
$\AssignSatisfiableLiterals(\boldF)$ always returns a satisfying assignment of $F$. 
Furthermore $p(\boldF,\alpha)$ defines a probability distribution on $\sat(\boldF)$.
If $n(\boldF)=0$, then $p(\boldF,\alpha)=1$. Otherwise we have the relation
\[p(\boldF,\alpha)=\frac{1}{|\SL(\boldF)|}\sum_{l\in\alpha}p(\boldF^{[l]},\alpha).\]
\end{observation}
Note that this distribution is not the uniform distribution: As an example consider the CNF formula corresponding to $x\vee y$. The probability that both $x$ and $y$ are set to $1$ is $1/4$, while the probability that exactly one of $x$ and $y$ is set to $0$ is $3/8$ each.

Using the probability distribution $p(\boldF,\alpha)$, we define a cost function on satisfiable $k$-CNF formulas. In the following fix an integer $s\geq 0$ and let $S:= S_k - \epsilon_k(s)$ s.t.\ $\pguessed(\boldF,x,\alpha,s)\leq S$ for all satisfiable $k$-CNF $\boldF$ where $x$ is frozen, as in Theorem \ref{theorem-guessed-bound}.
\begin{definition}
For a $(\leq k)$-CNF formula $\boldF$ with variable set $V(\boldF)$ we define the \emph{cost of $x$ in $\boldF$} as
\[c(\boldF,x):=\begin{cases}
0&\textrm{ if }x\not\in V(\boldF)\\
S&\textrm{ if }x\in \VN(\boldF)\\
\sum_{\alpha\in\sat(\boldF)}p(\boldF,\alpha)\pguessed(\boldF,x,\alpha,s)&\textrm{ if }x\in \VC(\boldF)\\
\end{cases}\]
We define the \emph{cost of $\boldF$} as $c(\boldF):=\sum_{x\in V(\boldF)}c(\boldF,x).$
\end{definition}
The cost of a variable that does not occur in the formula is set to $0$ for notational convenience. 
It follows from the definition that $c(\boldF,x)\leq S$ and hence $c(\boldF)\leq n(\boldF) S$. The cost function gives a lower bound on the success probability of PPSZ:
\begin{theorem}
\label{theorem-cost-bound}
$\psuccess(\boldF,s)\geq 2^{-c(\boldF)}.$
\end{theorem}
To obtain Theorems \ref{bound-3sat} and \ref{bound-4sat}, we choose $s$ such that $\epsilon_k(s)$ becomes small enough and $2^{S}<1.30704$ for $3$-SAT and $2^{S}<1.46899$ for $4$-SAT. By $O(2^{Sn})$ independent repetitions of PPSZ, the claimed randomized exponential algorithm can then be obtained by a routine argument. In the remainder of this section, we prove Theorem \ref{theorem-cost-bound}. We need the following lemma about $p(\boldF,\alpha)$:
\begin{lemma}
\label{lemma-probability-properties}
For $l\in\alpha$, we have $p(\boldF^{[l]},\alpha)\geq p(\boldF,\alpha)$. If $l$ is over a frozen variable, then $p(\boldF^{[l]},\alpha)=p(\boldF,\alpha)$, and $c(\boldF^{[l]})\leq c(\boldF)$.
\end{lemma}
\begin{proof} 
Consider
$\AssignSatisfiableLiterals(\boldF)$ given that $l$ is chosen at some point of time, let $\alpha'$ denote the output. The distribution of $\alpha'\setminus\{l\}$ is the same as the output of $\AssignSatisfiableLiterals(\boldF^{[l]})$, as is easily checked by induction. If $l$ is \emph{not} chosen in $\AssignSatisfiableLiterals(\boldF)$, then the output is never $\alpha$. Hence the probability that $\AssignSatisfiableLiterals(\boldF^{[l]})$ returns $\alpha\setminus\{l\}$ is at least the probability that $\AssignSatisfiableLiterals(\boldF)$ returns $\alpha$. This proves the first statement. If $l$ is over a frozen variable, then in $\AssignSatisfiableLiterals(\boldF)$ $l$ must be chosen at some point, and equality holds. The inequality on the costs now follows from Lemma \ref{lemma-pguessed-monotone}.
\end{proof}

If $\boldF$ is $s$-implication free the cost decreases by a certain amount depending on how many variables are frozen and non-frozen:
\begin{theorem}
\label{theorem-cost-decrease}
Suppose $\boldF$ is $s$-implication free.
For $l$ chosen u.a.r\ from $\SL(\boldF)$, we have
\[\E_l[c(\boldF^{[l]})]\leq c(\boldF)-\nN(\boldF)\frac{2S}{|\SL(\boldF)|}-\nC(\boldF)\frac{1}{|\SL(\boldF)|}.\]
\end{theorem}
If all variables are frozen, the cost decreases by $1$. If all variables are non-frozen, the cost decreases by $S<1$.
We will prove Theorem \ref{theorem-cost-decrease} later and use it now to prove Theorem \ref{theorem-cost-bound}:
\begin{proof}[Proof of Theorem \ref{theorem-cost-bound}]
We prove $\psuccess(\boldF,s)\geq 2^{-c(\boldF)}$ by induction on $n(\boldF)$.
If $n(\boldF)=0$, the statement is trivial. Assume the statement holds for formulas with less than $n(\boldF)$ variables, so that for $l\in\SL(\boldF)$, we have $\psuccess(\boldF^{[l]})\geq 2^{-c(\boldF^{[l]})}.$
If $\boldF$ is \emph{not} $s$-implication free, then let $l$ be the first $s$-implied literal fixed in $\PPSZ$ such that $\psuccess(\boldF)=\psuccess(\boldF^{[l]})$. The literal $l$ must be over a frozen variable, and from the last statement of Lemma \ref{lemma-probability-properties} we have $c(\boldF)\geq c(\boldF^{[l]})$ and hence $2^{-c(\boldF^{[l]})}\geq 2^{-c(\boldF)}$ and using the induction hypo\-the\-sis we are done.

Now assume that $\boldF$ is $s$-implication free.
Using Observation \ref{observation-induction} and the induction hy\-po\-the\-sis gives us
\[\psuccess(\boldF,s)=\frac{1}{2n(\boldF)}\sum_{l\in\SL(\boldF)}\psuccess(\boldF^{[l]},s)\geq\frac{1}{2n(\boldF)}\sum_{l\in\SL(\boldF)}2^{-c(\boldF^{[l]})}.\]
If we choose $l\in\SL(\boldF)$ u.a.r.\ we can write the sum as an expectation and then use Jensen's inequality and obtain
\[\psuccess(\boldF,s)\geq\frac{|\SL(\boldF)|}{2n(\boldF)}\E_{l}\left[2^{-c(\boldF^{[l]})}\right]\geq \frac{|\SL(\boldF)|}{2n(\boldF)}2^{-\E_{l}[c(\boldF^{[l]})]}=2^{\log(\frac{|\SL(\boldF)|}{2n(\boldF)})-\E_{l}[c(\boldF^{[l]})]}.\]
To prove the statement, we need to show that the exponent is at least $-c(\boldF)$, i.e.
\[L:=\log\left(\frac{|\SL(\boldF)|}{2n(\boldF)}\right)-\E_{l}[c(\boldF^{[l]})]+c(\boldF)\geq 0.\]
We bound the left-hand with Theorem \ref{theorem-cost-decrease} and obtain
\[L\geq \log\left(\frac{|\SL(\boldF)|}{2n(\boldF)}\right)-c(\boldF)+\nN(\boldF)\frac{2S}{|\SL(\boldF)|}+\nC(\boldF)\frac{1}{|\SL(\boldF)|}+c(\boldF)\]
\[=\log\left(\frac{|\SL(\boldF)|}{2n(\boldF)}\right)+\nN(\boldF)\frac{2S}{|\SL(\boldF)|}+\nC(\boldF)\frac{1}{|\SL(\boldF)|}\]
\[=\log\left(\frac{|\SL(\boldF)|}{n(\boldF)}\right)-1+\nN(\boldF)\frac{2S}{|\SL(\boldF)|}+\nC(\boldF)\frac{1}{|\SL(\boldF)|}.\]
Using twice $|\SL(\boldF)|=\nC(\boldF)+2\nN(\boldF)$, this is
\[=\log\left(1+\frac{\nN(\boldF)}{n(\boldF)}\right)+\nN(\boldF)\frac{2S}{|\SL(\boldF)|}-2\nN(\boldF)\frac{1}{|\SL(\boldF)|}.\]
With the inequality $\log(1+x)\geq \log(e)\frac{x}{1+x}$ (which is easily seen by writing $\log(1+x)$ as an integral), we have
\[L\geq\log(e) \frac{\frac{\nN(\boldF)}{n(\boldF)}}{\frac{|\SL(\boldF)|}{n(\boldF)}}+\nN(\boldF)\frac{2S}{|\SL(\boldF)|}-2\nN(\boldF)\frac{1}{|\SL(\boldF)|}\]
\[=\log(e) \frac{\nN(\boldF)}{|\SL(\boldF)|}-(2-2S)\frac{\nN(\boldF)}{|\SL(\boldF)|}.\]
It can be easily seen from the definition that $S_k$ increases for larger $k$. Hence $S\geq S_3=2\ln 2 -1\approx 0.3863$ and 
$(2-2S)\leq 4- 4\ln 2 < 1.23 < 1.44 < \log(e)$, 
which implies $L\geq 0$ and completes the proof.
\end{proof}
It is interesting to see that we still have some leeway in the last step. One checks that $\log(e)=(2-2\frac{S_3}{1+S_3})$ which means that our method works as long the upper bound $S$ that a frozen variable is guessed is at least $\frac{S_3}{1+S_3}\approx 0.2787$, corresponding to an algorithm with running time roughly $O(1.214^n)$.
\subsection{Remaining Proofs}
We now need to prove Theorem \ref{theorem-cost-decrease}. The theorem follows from the following two lemmas:
\begin{lemma}
\label{lemma-noncrit-decrease}
If $x\in \VN(\boldF)$, then for $l$ chosen u.a.r\ from $\SL(\boldF)$ we have
\[\E_l[c(\boldF^{[l]},x)]\leq c(\boldF,x)-\frac{2S}{|\SL(\boldF)|}.\]
\end{lemma}
Note that this lemma holds even if $\boldF$ has $s$-implied literals. However, we only use it if $\boldF$ is $s$-implication free.
\begin{proof}
$x$ is non-frozen, so by definition $c(\boldF,x)=S$. As the cost of any variable is at most $S$, we have
\[E_l[c(\boldF^{[l]},x)]\leq S \Pr_l[x\in V(\boldF^{[l]})]=S \Pr_l[l\textrm{ is not over }x]=S \frac{|\SL(\boldF)|-2}{|\SL(\boldF)|}.\]
\end{proof}
\begin{lemma}
\label{lemma-crit-decrease}
If $x\in \VC(\boldF)$ and not $s$-implied, then for $l$ chosen u.a.r\ from $\SL(\boldF)$ we have
\[\E_l[c(\boldF^{[l]},x)]\leq c(\boldF,x)-\frac{1}{|\SL(\boldF)|}.\]
\end{lemma}
\begin{proof}
By writing the expectation as a sum and inserting the definition, we have
\[\E_l[c(\boldF^{[l]},x)]=\frac{1}{|\SL(\boldF)|}\sum_{l\in\SL(\boldF)}\sum_{\alpha'\in\sat(\boldF^{[l]})}p(\boldF^{[l]},\alpha')\pguessed(\boldF^{[l]},x,\alpha',s).\]
Note that we have extended $p$ and $\pguessed$ such that also assignments over a set $W\supset V(\boldF^{[l]})$ are allowed and the additional variables are ignored.
We now claim that we can exchange the sums and get
\[\E_l[c(\boldF^{[l]},x)]=\frac{1}{|\SL(\boldF)|}\sum_{\alpha\in\sat(\boldF)}\sum_{l\in\alpha}p(\boldF^{[l]},\alpha)\pguessed(\boldF^{[l]},x,\alpha,s).\]
To prove this, we need to show that there is a bijection between the set $\{(\alpha,l)\ |\ \alpha\in\sat(\boldF),l\in\alpha\}$ and the set $\{(l,\alpha')\ |\ l\in SL(\boldF),\alpha'\in \boldF^{[l]}\}$ s.t.\ $\alpha=\alpha'$ on $V(\boldF^{[l]})$. One can easily check that $f(\alpha,l):=(l,\alpha\setminus \{l\})$ with $f^{-1}(l,\alpha')=(\alpha'\cup \{l\},l)$ is such a bijection.

Now the outer sum is over $\alpha\in\sat(\boldF)$, as in the definition of $c(\boldF,x)$. Hence it is sufficient to prove that for all $\alpha\in\sat(\boldF)$ we have
\[\frac{1}{|\SL(\boldF)|}\sum_{l\in\alpha}p(\boldF^{[l]},\alpha)\pguessed(\boldF^{[l]},x,\alpha,s)\leq p(\boldF,\alpha)\pguessed(\boldF,x,\alpha,s)-p(\boldF,\alpha)\frac{1}{|\SL(\boldF)|}.\]
We multiply this by $\SL(\boldF)$ and divide it by $p(\boldF,\alpha)$ (which is trivially positive) and get equivalently
\begin{equation}
\label{eq-texpression}
T:=\sum_{l\in\alpha}\frac{p(\boldF^{[l]},\alpha)}{p(\boldF,\alpha)}\pguessed(\boldF^{[l]},x,\alpha,s)\leq |\SL(\boldF)|\pguessed(\boldF,x,\alpha,s)-1.
\end{equation}
It remains to show (\ref{eq-texpression}). We bound the left-hand side $T$. First we split it into two sums:
\[T=\sum_{l\in\alpha}\pguessed(\boldF^{[l]},x,\alpha,s)+\sum_{l\in\alpha}\frac{p(\boldF^{[l]},\alpha)-p(\boldF,\alpha)}{p(\boldF,\alpha)}\pguessed(\boldF^{[l]},x,\alpha,s).\]
Now we use Lemma \ref{lemma-pguessed-reduced} multiplied by $n(\boldF)$ on the first sum. By Lemma \ref{lemma-probability-properties} all summands of the second sum are positive, so we can use Lemma \ref{lemma-pguessed-monotone} to bound $\pguessed(\boldF^{[l]},x,\alpha,s)$ from above by $\pguessed(\boldF,x,\alpha,s)$. We obtain
\[T\leq n(\boldF)\pguessed(\boldF,x,\alpha,s)-1+\sum_{l\in\alpha}\frac{p(\boldF^{[l]},\alpha)-p(\boldF,\alpha)}{p(\boldF,\alpha)}\pguessed(\boldF,x,\alpha,s).\]
Observation \ref{observation-distribution} tells us that $\sum_{l\in\alpha}p(\boldF^{[l]},\alpha)=|\SL(\boldF)|p(\boldF,\alpha)$, and so
\[\sum_{l\in\alpha}\frac{p(\boldF^{[l]},\alpha)-p(\boldF,\alpha)}{p(\boldF,\alpha)}=|\SL(\boldF)|-n(\boldF).\]
Therefore
\[T\leq n(\boldF)\pguessed(\boldF,x,\alpha,s)-1+(|\SL(\boldF)|-n(\boldF))\pguessed(\boldF,x,\alpha,s),\]
which is equal to $|\SL(\boldF)|\pguessed(\boldF,x,\alpha,s)-1$ and hence (\ref{eq-texpression}) holds.
\end{proof}

Theorem \ref{theorem-cost-decrease} can now be easily proved:
\begin{proof}[Proof of Theorem \ref{theorem-cost-decrease}]
We need to show that $\boldF$ is $s$-implication free and $l$ chosen u.a.r\ from $\SL(\boldF)$, we have
\[\E_l[c(\boldF^{[l]})]\leq c(\boldF)-\nN(\boldF)\frac{2S}{|\SL(\boldF)|}-\nC(\boldF)\frac{1}{|\SL(\boldF)|}.\]
Using the definition of the cost we obtain
\[\E_l[c(\boldF^{[l]})]=\E_l\left[\sum_{x\in V(\boldF^{[l]})}c(\boldF^{[l]},x)\right]=\E_l\left[\sum_{x\in V(\boldF)}c(\boldF^{[l]},x)\right].\]
Then linearity of expectation gives
\[\E_l[c(\boldF^{[l]})]=\sum_{x\in V(\boldF)}\E_l[c(\boldF^{[l]},x)].\]
Now we can plug in Lemmas \ref{lemma-noncrit-decrease} and \ref{lemma-crit-decrease} to get
\[\E_l[c(\boldF^{[l]})]\leq \sum_{x\in V(\boldF)}c(\boldF,x)-\nN(\boldF)\frac{2S}{|\SL(\boldF)|}-\nC(\boldF)\frac{1}{|\SL(\boldF)|}.\]
Using the definition of $c(\boldF)$ gives the statement.
\end{proof}
\section{Conclusion}
We have shown an analysis of a slightly adapted PPSZ algorithm that gives the same bound for general $k$-SAT as for Unique $k$-SAT. For $k\geq 5$, this was already known, but our analysis might be considered more intuitive. For $k=3$ and for $k=4$ this gives improved running time bounds; for $k=3$ the bound significantly improves from $O(1.32065^n)$ to $O(1.30704^n)$. The fastest known randomized algorithm for $3$-SAT is now again rather simple compared with the algorithm proposed in~\cite{hms10}. It is noteworthy that this is the first algorithm for $3$-SAT that is faster than, but independent of Sch\"oning's algorithm~\cite{schoning1999}. The best known bounds for Unique $k$-SAT and $k$-SAT match now, but it is still an open question if this holds in general, as conjectured by Calabro et al.~\cite{cikp08}. 

For \emph{deterministic} algorithms, the picture is a bit different. Recently Sch\"oning's algorithm has been fully derandomized by Moser and Scheder~\cite{schderandom} yielding a deterministic algorithm for $3$-SAT running in time $O(1.33334^n)$. This has been improved very recently by Makino, Tamaki, and Yamamoto~\cite{mty11} to $O(1.3303^n)$. For Unique $k$-SAT, PPSZ has been fully derandomized by Rolf in 2005~\cite{rolfderandom}, giving a deterministic algorithm for Unique $3$-SAT running in time $O(1.30704^n)$, as the randomized version. Our new approach to generalize from Unique $k$-SAT to $k$-SAT in PPSZ might be used to derandomize PPSZ for general $k$-SAT. 

We have adapted PPSZ slightly by immediately using $s$-implied literals. In the original PPSZ, $s$-implied variables are good because they behave like non-frozen variables in the sense that restricting to them preserves satisfiability. However, while non-frozen variables still have an expected cost reduction of $\frac{2S}{|\SL(\boldF)|}\approx \frac{0.773}{|\SL(\boldF)|}$, the cost reduction of $s$-implied variables is $0$, as they are guessed with probability $0$ and hence already have cost $0$. Our approach needs cost reduction at least $\frac{2-\log(e)}{\SL(\boldF)|}\approx \frac{0.557}{|\SL(\boldF)|}$. It might be interesting to check if our approach can be improved to overcome this problem and accommodate the original PPSZ algorithm.
\section*{Acknowledgements}
I am very grateful to Heidi Gebauer, Dominik Scheder, and Emo Welzl for checking my ideas. Special thanks go to Robin Moser for continuous assistance in realizing this paper.
\bibstyle{plain}

\bibliography{../common}

\end{document}